\documentclass[a4paper,12pt]{article}
\usepackage[colorlinks=true,urlcolor=blue,linkcolor=red,citecolor=blue,anchorcolor=blue,breaklinks]{hyperref}
\usepackage{amsmath,amssymb,amsthm}
\usepackage{sectsty}
\sectionfont{\large}
\usepackage{fancyhdr}
\newcounter{theorem}
\newtheorem{Theorem}[theorem]{Theorem}
\newcounter{corollary}
\newtheorem{Corollary}[corollary]{Corollary}
\newcounter{definition}
\newtheorem{Definition}[definition]{Definition}
\newcounter{remark}
\newtheorem{Remark}[remark]{Remark}

\begin{document}
\sloppy

\parindent 0cm

\begin{center}
\textbf{\Large{Chirality in Affine Spaces and in Spacetime}}
\end{center}
\vskip 0.5cm

\begin{center}
\textbf{Michel Petitjean}
\end{center}
\vskip 0.5cm

\small{$^{1}$ Universit\'e de Paris, BFA, CNRS UMR 8251, INSERM ERL U1133, F-75013 Paris, France.\\
       $^{2}$ E-p\^ole de G\'enoinformatique, CNRS UMR 7592, Institut Jacques Monod, 75013 Paris, France.}\\
       \textit{E-mail}: petitjean.chiral@gmail.com, michel.petitjean@u-paris.fr.\\
       \href{http://petitjeanmichel.free.fr/itoweb.petitjean.html}{http://petitjeanmichel.free.fr/itoweb.petitjean.html}.\\
       \textit{ORCID}: \href{https://orcid.org/0000-0002-1745-5402}{0000-0002-1745-5402}.

\vskip 0.5cm
\textbf{Abstract}:
An object is chiral when its symmetry group contains no indirect isometry.
It can be difficult to classify isometries as direct or indirect, except in the Euclidean case.
We classify them with the help of outer semidirect products of isometry groups, in particular in the case of an affine space defined over a finite-dimensional real quadratic space.
We also classify as direct or indirect the isometries of the real Lorentz-Minkowski spacetime and those of the classical spacetime defined by the Newton-Cartan theory.



\parindent 1cm
\section{Introduction}\label{sec:INTRO}

Chirality is of major importance in physics, chemistry, and biology \cite{Kamenetskii2021,Mezey1991,Hargittai2009,Palyi2020}.
Euclidean isometries are usually classified into direct and indirect isometries.
An object is chiral if and only if its symmetry group contains no indirect isometry (the term object is defined in Section \ref{sec:DICO}).
A linear isometry of $E^n$ is represented by a $n \times n$ matrix.
This isometry is direct if the matrix determinant is $+1$, and it is indirect if this determinant is $-1$ \cite{Cederberg2001}.
It follows that a reflection is an indirect isometry and that any composition of linear isometries containing an odd number of reflections is an indirect isometry.
The Euclidean translations are also direct isometries, and it was stated that Euclidean direct isometries are those preserving orientation \cite{Cederberg2001}.
According to the historical definition of chirality coined by Lord Kelvin \cite{Kelvin1894}, a set embedded in $E^n$ is chiral if no rigid motion of this set can bring it into superposition with its mirror image, and otherwise it is said to be achiral \cite{Mezey1997}.
It was also stated that the absence of any kind of reﬂection symmetry is the prerequisite for chirality \cite{Gerlach2009}.

But neither computing matrix determinants nor looking at orientation \mbox{preservation} is relevant to define direct and indirect isometries in general spaces.
A suitable \mbox{definition} of direct and indirect isometries is required.
We retain one which is based only on the group of space isometries (so that no space orientation is needed), and which recovers the usual meaning in $E^n$ \cite{Petitjean2020}.
It is presented in Section \ref{sec:DICO}.
In \mbox{Section \ref{sec:RUL}} we present the composition rules of isometries in the general case.
Our main results about classication of isometries are in Section \ref{sec:SDP}.
In \mbox{Section \ref{sec:AFF}} we use these results to make explicit the composition rules for the affine space over a finite-dimensional real quadratic space.
We also make explicit these rules for the real Lorentz-Minkowski spacetime in Section \ref{sec:MINK}, and for the Newton-Cartan spacetime (that is, the classical spacetime defined by the Newton-Cartan theory), in Sections \ref{sec:GAL} and \ref{sec:FGAL}.
In section \ref{sec:CONE}, we consider the example of an object in the the Newton-Cartan spacetime, the rotating cone.

\section{Direct and indirect isometries; chiral and achiral objects}\label{sec:DICO}

It is mandatory to know which mathematical entities can be qualified as symmetric or asymmetric.
These entities are called objects.

\begin{Definition}\label{def:D1}
An object is a function which has its input argument in a metric space. \cite{Petitjean2007}
\end{Definition}

The domain of this function is a metric space, and the space $B$ of its image values has to be defined by the user.
For example, when the metric space is $E^n$, the object can be the indicator function of a domain of $E^n$ and in this case $B=\lbrace 0;1 \rbrace$.
Several other examples of objects are given in \cite{Petitjean2007}, in the case of probability distributions, graphs, matrices, colored figures, and strings.

\begin{Definition}\label{def:D2}
An object is symmetric when it is invariant under an isometry which is not the identity. \normalfont{\cite{Petitjean2007}}
\end{Definition}

All throughout the text, the composition of an isometry with itself is called a squared isometry.
Definition \ref{def:D2} is the usual one.
The isometry group must be exhibited.

\begin{Definition}\label{def:D3}
An isometry is direct when it can be expressed as a product of squared \mbox{isometries \normalfont{\cite{Petitjean2020}}}.
An isometry which is not direct is an indirect isometry.
\end{Definition}

Definition \ref{def:D3} is not the usual one in vigor in $E^n$, but it recovers this usual \mbox{definition \cite{Petitjean2020}}.
There is no need to deal with space orientation.
Any isometry \mbox{group $G$} is the union of its subgroup $G^D$ of the direct isometries and of the complement $G^I=G \setminus G^D$ of the indirect isometries \cite{Petitjean2020}. 
The group $G^D$ contains at least the identity, and $G^I$ may be empty or not.

\begin{Definition}\label{def:D4}
An object is achiral when it is invariant under an indirect isometry.
If its symmetry group contains no indirect isometry, the object is chiral.
\end{Definition}

Definition \ref{def:D4} is the usual one, but it relies upon Definition \ref{def:D3}.
Orientability, when existing, is a property of the space.
It may be global or local.
The potential existence of chiral objects is also a property of the space, which may also stand globally or \mbox{locally \cite{Petitjean2021a}}.
Both properties depend on the space structure, but none of them is the cause of the existence of the other one.

\begin{Remark}\label{rem:R1}
None of the requirements defining the metric are needed in Definition \ref{def:D1} \normalfont{\cite{Petitjean2007}}.
\\
It follows that the four definitions \ref{def:D1}-\ref{def:D4} work even when the metric is not a true one.
\end{Remark}

This is useful in physics.
For example, in the Lorentz-Minkowski spacetime, intervals are preserved rather than distances, and the objects can be scalar fields, vector fields, or else.
More generally, intervals are preserved in quadratic spaces.
Direct and indirect isometries of the inhomogeneous Lorentz group were exhibited \cite{Petitjean2019}.
They were also exhibited for the orthogonal group in finite-dimensional real quadratic spaces \cite{Petitjean2021b}.
The symmetry invariance may be based neither on a metric nor on a pseudometric.
This is the case for the Newton-Cartan spacetime, where the invariance relies on two degenerate metrics \cite{Andringa2011}.

\section{Rules of composition of isometries}\label{sec:RUL}

The following general rules of composition of isometries are a consequence of Definition \ref{def:D3}. 

\begin{Theorem}\label{th:RUL}
(R1) The product of two direct isometries is a direct isometry.
\\
(R2) The product of a direct isometry by an indirect isometry is an indirect isometry.
\\
(R3) The product of an indirect isometry by a direct isometry is an indirect isometry.
\\
(R4) The product of two indirect isometries is either a direct isometry or an indirect isometry.
\end{Theorem}

\begin{proof}
See \cite{Petitjean2020}.
\end{proof}

Rules \textit{R1}, \textit{R2} and \textit{R3} are those in vigor in $E^n$, while rule \textit{R4} must be explicited, depending on which space is considered.
In $E^n$, the product of two indirect isometries is always a direct isometry. 
Rule \textit{R4} was explicited for the inhomogeneous Lorentz group \cite{Petitjean2019}, and for the orthogonal group in finite-dimensional real quadratic spaces \cite{Petitjean2021b}.

\section{Semidirect products of isometry groups}\label{sec:SDP}

We do not yet emit any assumption about the space, apart from having an isometry group, and from this group satisfying the conditions stated in the theorems given in this section.
The symmetry invariance is not assumed to rely on a true metric.

For clarity in the statements of the theorems and by abuse of language, when we build an isometry group $G$ as an outer (that is, external) semidirect product of its subgroups, we do not distinguish an isometry of a subgroup of $G$ for the group operation in this subgroup, from this isometry in $G$ for the group operation in $G$ (this distinction is done in the proofs, however).


\begin{Theorem}\label{th:SDPD}
Let $G = N \rtimes H$ be an isometry group constructed as an outer semidirect product of a normal subgroup $N$ and a subgroup $H$. 
\\
(a) Any direct isometry of $H$ is a direct isometry of $G$.
\\
(b) Any direct isometry of $N$ is a direct isometry of $G$.
\end{Theorem}

\begin{proof}
We denote by $\odot$ the group operation in $N$.
No symbol is used to denote the group opration in $H$, and no symbol is used to denote the group operation in $G$, the writing of products will be clear from the context.
Let $\phi_{h}$ be an homomorphism from $H$ to the automorphism group of $N$, such that, for all $n \in N$ and all $h \in H$, $\phi_{h}(n)$ is the conjugate of $n$ by $h$, that is, $\phi_{h}(n)=hnh^{-1}$.
Thus, for all $(n_1,h_1) \in G$ and all $(n_2,h_2) \in G$, we have $(n_1,h_1)(n_2,h_2)=(n_1 \odot \phi_{h_1}(n_2) , h_1 h_2)$ \cite{Hall1959,Robinson2003}. 
We denote by $(e_N,e_H)$ the neutral element in $G$.

(a) We consider a direct isometry $h_D$ of $H$.
From Definition \ref{def:D3}, $h_D$ is a product of squared elements of $H$.
Then, the elements $(e_N,h)$ of $G$ form a subgroup isomorphic to $H$ (see Theorem 6.5.2 in \cite{Hall1959}).
Thus, we deduce that $(e_N,h_D)$ can be written as a product of squared elements \mbox{of $G$}.

(b) We consider a direct isometry $n_D$ of $N$.
From Definition \ref{def:D3}, $n_D$ is a product of squared elements of $N$.
Then, the elements $(n,e_H)$ of $G$ form a normal subgroup isomorphic to $N$ (see Theorem 6.5.2 in \cite{Hall1959}).
Thus, we deduce that $(n_D,e_H)$ can be written as a product of squared elements of $G$.
\end{proof}

\begin{Theorem}\label{th:SDPI}
Let $G = N \rtimes H$ be an isometry group constructed as an outer semidirect product of a normal subgroup $N$ and a subgroup $H$. 
Any indirect isometry of $H$ is an indirect isometry of $G$.
\end{Theorem}

\begin{proof}
(a) We consider an indirect isometry $h_I \in H$ and we recall that the elements $(e_N,h)$ of $G$ form a subgroup isomorphic to $H$.
We assume that $(e_N,h_I$) can be written as a product of squared elements of $G$, that is, $(e_N,h_I)=\prod\limits_k (n_k \odot \phi_{h_k}(n_k) , h_k^2)$, where $(n_k,h_k) \in G$ and $\phi_{h_k}$ is the homorphism associated to the product $(n_k,h_k)^2$.
Then, we can write $(e_N,h_I)=(t,\prod\limits_k h_k^2)$, where $t \in N$, from which we deduce that $t=e_N$ and $h_I=\prod\limits_k h_k^2$, a contradiction.
\end{proof}

\begin{Theorem}\label{th:DPI}
Let $G = N \times H$ be an isometry group constructed as the direct product of a normal subgroup $N$ and a normal subgroup $H$, such that their intersection is the neutral element of $G$. 
\\
(a) Any indirect isometry of $H$ is an indirect isometry of $G$.
\\
(b) Any indirect isometry of $N$ is an indirect isometry of $G$.
\end{Theorem}

\begin{proof}
The direct product $N \times H$ is a special case of the semidirect product $N \rtimes H$, where the homomorphism $\phi_{h}$ returns $\phi_{h}(n)=n$ for all $n \in N$ and all $h \in H$ \cite{Robinson2003}.
\\
(a) Apply Theorem \ref{th:SDPI}.
\\
(b) Permute the roles of $H$ and $N$ and apply the proof of part (a).
\end{proof}

\section{Isometries in affine spaces}\label{sec:AFF}

First, we consider a finite-dimensional real vector space $V$ endowed with a non-degenerate quadratic form $Q$ of signature $(p,q)$.
The orthogonal group $O(p,q)$ acting on the quadratic space $(V,Q)$ contains the reflections and their products, and it acts linearly on $(V,Q)$. 
The members of $O(p,q)$ are isometries which preserve intervals.
The interval $S_{x,y}$ between two vectors $x$ and $y$ of $V$ is defined by $S_{x,y}^2=Q(x-y)$ (this quantity may be positive, negative or null).
The isometries of $O(p,q)$ are classified with Theorem \ref{th:CLASONR}, where the vector othogonal to a reflection hyperplane is called the supporting vector of this reflection.

\begin{Theorem}\label{th:CLASONR} 
A product of reflections of the orthogonal group $O(p,q)$ acting on $(V,Q)$ is classified as direct or indirect as follows.
\\
(a) When the product contains an odd number of reflections, it is an indirect isometry.

When the product contains an even number of reflections:
\\
(b) It is a direct isometry when an even number of the squares of the supporting vectors have a negative sign, and the squares of the other supporting vectors have a positive sign.
\\
(c) It is an indirect isometry when an odd number of the squares of the supporting vectors have a negative sign, and the squares of the other supporting vectors have a positive sign.
\end{Theorem}

\begin{proof}
See Theorem 13 in \cite{Petitjean2021b}.
\end{proof}

No reflection is supported by a vector of null square \cite{Petitjean2021b}.
Then, a single \mbox{reflection} is an indirect isometry (Theorem \ref{th:CLASONR}a), and $O(p,q)$ is generated by the reflections (see Cartan-Dieudonn\'e theorem \cite{Artin1957,Aragon2006}).
It follows that Theorem \ref{th:CLASONR} suffices to classify all isometries of $O(p,q)$ as direct or indirect.

We wish to extend the validity of Theorem \ref{th:CLASONR} to the case of an affine space \mbox{over $(V,Q)$}.
So, we must consider the translations, because the intervals are also preserved by the members of the translation group $\mathcal{T}$ \mbox{operating} on $(V,Q)$, that is, for all vectors $x \in V$ and $y \in V$ and for all translations $\tau_t \in \mathcal{T}$ of vector $t \in V$, $Q((x+t)-(y+t))=Q(x-y)$.
\mbox{The group} operation of $\mathcal{T}$ is the addition and it is \mbox{commutative}.

We show in Theorem \ref{th:TRANS} that translations are direct isometries, then we extend the validity of Theorem \ref{th:CLASONR} in Theorem \ref{th:CLASONRT}.

\begin{Theorem}\label{th:TRANS}
All isometries of $\mathcal{T}$ are direct.
\end{Theorem}

\begin{proof}
It is known that two finite dimensional vector spaces over the same field are isomorphic if and only if they have the same dimension (see Theorem 1.7D in \cite{Thrall2014}). 
Then, the translation group of an affine space is isomorphic to the additive group of the underlying vector space (see Theorem 6 in \cite{Birkhoff1977}). 
Thus, the set $\mathcal{T}$ is isomorphic to $\mathbb{R}^{p,q}$ for the addition.
It follows that any translation $\tau \in \mathcal{T}$ can be written $\tau=(\tau/2)+(\tau/2)$, which means that $\tau$ is the composition of $\tau/2$ with itself.
In other words, any translation is always the square of an other translation, so it is a direct isometry.
\end{proof}

\begin{Theorem}\label{th:CLASONRT}
Let $A$ be the affine space over the finite-dimensional real quadratic space $(V,Q)$, \mbox{$Q$ being} a non-degenerate quadratic form of signature $(p,q)$, and let $\mathcal{T}$ be the translation group acting on $A$. 
The isometries of the orthogonal group $O(p,q)$ acting on $A$ are classified according to the rules (a), (b) and (c) of Theorem \ref{th:CLASONR}, the translations of $\mathcal{T}$ are direct isometries, and the compositions of isometries of $O(p,q)$ with translations are also classified according to the rules (a), (b) and (c) of Theorem \ref{th:CLASONR}.
\end{Theorem}

\begin{proof}
We consider the semidirect product $G = \mathcal{T} \rtimes O(p,q)$. 
It is an affine group preserving intervals. 

Let $e_0$ be the neutral element of $\mathcal{T}$ and $e_1$ be the neutral element of $O(p,q)$.
\mbox{According} to Theorem 6.5.2 in \cite{Hall1959}, the elements $(e_0,h)$ form a subgroup of $G$, $O_G(p,q)$, which isomorphic to $O(p,q)$, and the elements $(\tau,e_1)$ form a normal subgroup $\mathcal{T}_G$ of $G$, which is isomorphic to $\mathcal{T}$, with $\mathcal{T}_G \cap O_G(p,q) = (e_0,e_1)$.

The normality of $\mathcal{T}_G$ stands for the following reason.
The elements of $O_G(p,q)$ act linearly on the vectors of $V$, so that, for all $(h,e_1) \in O_G(p,q)$, for all translations $(e_0,\tau) \in \mathcal{T}_G$ and for all $x \in V$, $((h,e_1) (e_0,\tau) (h,e_1)^{-1}) (x) = x + (e_0,h)(t)$, where $t \in V$ is the vector bijectively associated to $(e_0,\tau)$, and $h(t)$ is constant in $V$, that is, $h(t)$ does not depend on $x$. 

Then, to extend the validity of Theorems \ref{th:CLASONR} and \ref{th:TRANS}, apply Theorems \ref{th:SDPD} and \ref{th:SDPI}.
Finally, due to rules (R2) and (R3) of Theorem \ref{th:RUL}, the composition on the left or on the right of linear isometries with translations still follow the rules of Theorem \ref{th:CLASONR}.
\end{proof}

The isometries of $A$ being generated by the translations and by the reflections, it follows that Theorem \ref{th:CLASONRT} suffices to classify all isometries of $A$ as direct or indirect.

\begin{Corollary}\label{cor:CLASONRTE}
In the case of the Euclidean space $E^n$, the rules of composition of isometries of Theorem \ref{th:RUL} are completed such that the composition of two indirect isometries is always a direct isometry.
\end{Corollary}

\begin{proof}
The affine space $E^n$ is isomorphic to $A$ in Theorem \ref{th:CLASONRT} when $q=0$ in the signature of the non-degenerate quadratic form $Q$.
This latter is positive definite, it exists no reflection with a supporting vector of negative square, and case (c) of Theorem \ref{th:CLASONR} cannot occur.
Translations are still direct isometries, as in the general case where $q>0$.
\end{proof}

\section{The real Lorentz-Minkowski spacetime and the Poincar\'e group}\label{sec:MINK}

The Poincar\'e group, also called the inhomogeneous Lorentz group, is the largest isometry group of Lorentz-Minkowski spacetime. 
It is the semidirect product of the Abelian translation group (in time and space) and of the Lorentz group \cite{Carmeli2000,Muller-Kirsten2010}.
The Lorentz group is a four components Lie group which is generated by spatial rotations, parity inversion $P$ (which turns all spatial coordinates into their opposite), time reflection $T$ (time reversal), and boosts, which are transformations connecting two uniformly moving bodies. 
The Minkowski metric, which is a pseudo-metric, is diagonal.
Depending on the authors, it is induced by a quadratic form of signature either $(1,3)$ or $(3,1)$, so that the Poincaré group is either $\mathbb{R}^{1,3} \rtimes O(1,3)$ \cite{Carmeli2000,Muller-Kirsten2010}, or $\mathbb{R}^{3,1} \rtimes O(3,1)$ \cite{McCabe2007}, while this difference is purely conventional in physics.

The semidirect product structure of the Poincaré group permits to classify its isometries as direct or indirect according to Theorem \ref{th:CLASONRT}, with either $(p,q)=(1,3)$ or $(p,q)=(3,1)$.
From Theorem \ref{th:SDPD}, except the involutions $P$ and $T$, the generators of the Lorentz group are direct isometries of the Poincar\'e group.
\mbox{In particular}, Lorentz boosts are direct isometries of the Lorentz group (\mbox{Corollary 1} \mbox{in \cite{Petitjean2019}}), so they are also direct isometries of the Poincar\'e group.
From Theorem \ref{th:SDPI}, the single reflections $P$ and $T$ are indirect isometries of the Poincar\'e group.

An interesting consequence is that the involution $PT=TP$ is an indirect isometry of the Poincar\'e group (case (c) of \mbox{Theorem \ref{th:CLASONR}}).
This classification of $PT$ as an indirect isometry was previously established using a matrix representation of the Lorentz group \cite{Petitjean2019}, and so we retrieve that $PT$ is an indirect isometry of the Poincar\'e group with Theorem \ref{th:SDPI}.
Both $P$ and $T$ are indirect isometries, but considering the full reflection $PT$ as a direct isometry would be an erroneous conclusion because the Lorentz-Minkowski spacetime is not Euclidean.
Furthermore, in the context of a chirality analysis, there is no reason to consider that the combination of a mirroring in space with a mirroring in time should be classified like a rotation. 

\section{The Newton-Cartan spacetime and the inhomogeneous Galilean group}\label{sec:GAL}

In Newton-Cartan theory, the spacetime is equipped with two diagonal degenerate metrics: a spatial metric of diagonal elements $(0,1,1,1)$ and a temporal metric of diagonal elements $(1,0,0,0)$ \cite{Andringa2011,Callender2017}.
In fact, there is no place for defining a four-dimensional metric \cite{Leihkauf1989}.
When a spatial distance is not null, the temporal distance is null, and when the temporal distance is not null, the spatial distance is null \cite{Callender2017}.
Speaking about the distance between two spatial points at different times is meaningless, and speaking about the time interval between two distinct spatial points is meaningless, too.
So, the Newton-Cartan spacetime is not a quadratic space.

The isometry group of the Newton-Cartan spacetime is the Galilean group $Gal(1,3)$.
It contains space and time translations, spatial rotations and boosts, also called pure Galilean transformations.
The homogeneous Galilean group $HGal(1,3)$ is generated by spatial rotations and boosts.
Galilean boosts differ from Lorentz boosts.
The composition of Galilean boosts is commutative and it is a Galilean boost \cite{deMontigny2006}, while the composition of Lorentz boosts is in general not a Lorentz boost because it is equivalent to a Lorentz boost preceded by a spatial rotation \cite{Ungar1988}.

\begin{Theorem}\label{th:CLASONRNC1}
All isometries of $Gal(1,3)$ are direct.
\end{Theorem}

\begin{proof}
Spatial rotations are direct isometries because a rotation is always the square of two rotations \cite{Petitjean2020}.

We denote by $(t,x)$ an element of the Newton-Cartan spacetime.
A Galilean boost $b_v$ of direction $v$, where $v \in \mathbb{R}^3$, transforms $(t,x)$ into $(t,x+tv)$ \cite{deMontigny2006}.
A Galilean boost $b_v$ is always equal to the composition of $b_{v/2}$ with itself, and so, according to \mbox{Definition \ref{def:D3}}, $b_v$ is a direct isometry.

The composition of Galilean boosts defines an Abelian subgroup of $HGal(1,3)$, and $HGal(1,3)$ is a semidirect product of the Abelian group of Galilean boosts and of the spatial rotation group \cite{deMontigny2006}.
So, from Theorem \ref{th:SDPD}, rotations and boosts are direct isometries of $HGal(1,3)$.

Translations in space and time are direct isometries (the proof is similar to the one of Theorem \ref{th:TRANS}).
The inhomogeneous Galilean group $Gal(1,3)$ is a semidirect product of the Abelian space and time translation group and of the homogeneous Galilean group \mbox{$HGal(1,3)$ \cite{deMontigny2006}}.
So, from Theorem \ref{th:SDPD}, spatial rotations and boosts are direct isometries of $Gal(1,3)$, and it follows from Theorem \ref{th:RUL} that all isometries of $Gal(1,3)$ are direct.
\end{proof}

\section{The full isometry group of the Newton-Cartan spacetime}\label{sec:FGAL}

We observe that $P$ and $T$ leave invariant the distances defined by the two degenerate metrics mentioned in section \ref{sec:GAL}.
So, we must include $P$ and $T$ as generators of the full isometry group of the Newton-Cartan spacetime, that we denote by $GAL(1,3)$.
It is generated by parity inversion $P$ and time reversal $T$, by the set $\mathcal{R}$ of spatial rotations, the set $\mathcal{B}$ of boosts, and the set $\mathcal{T}$ of space and time translations.
We build $GAL(1,3)$ via successive inclusions of its subset of isometries.

\begin{Theorem}\label{th:GALK}
The set $\mathcal{K}=\lbrace I,P,T,PT \rbrace$, where $I$ is the identity, is an Abelian group isomorphic to the Klein four-group, that is, the direct product of two copies of the cyclic group of order $2$.
Except $I$, it contains only indirect isometries, which are $P$, $T$ and $PT$.
\end{Theorem}

\begin{proof}
For all vectors $(t,x)$, $(TP)(t,x)=(PT)(t,x)=(-t,-x)$, so $PT=TP$.
Both $P$ and $T$ are involutions, therefore $PT=TP$ is also an involution.
Then it can be checked that $\mathcal{K}$ is an Abelian group.
The isomorphism of $\mathcal{K}$ with the Klein four-group can be deduced from their Cayley tables.

All elements of $\mathcal{K}$ have a square equal to $I$, so, except $I$ itself, none of them can be written as product of squares in $\mathcal{K}$.
\end{proof}

\begin{Theorem}\label{th:GALQ}
(a) The isometry group $\mathcal{Q} = \mathcal{K} \cup \mathcal{R}$ is the direct product of $\mathcal{K}$ and $\mathcal{R}$, \mbox{that is}, \mbox{$\mathcal{Q} = \mathcal{K} \times \mathcal{R}$}.
\\
(b) The direct symmetries of $\mathcal{R}$ are direct symmetries of $\mathcal{Q}$.
\\
(c) The indirect symmetries of $\mathcal{K}$ are indirect symmetries of $\mathcal{Q}$.
\end{Theorem}

\begin{proof}
(a) $\mathcal{K}$ is a group (Theorem \ref{th:GALK}) and so it is a subgroup of $\mathcal{Q}$, $\mathcal{R}$ is a group (see \mbox{Section \ref{sec:GAL}}) and so it a subgroup of $\mathcal{Q}$, and the intersection of $\mathcal{K}$ and $\mathcal{R}$ is the identity.

For all $r \in \mathcal{R}$ and for all vectors $(t,x)$, we observe that $(Pr)(t,x)=(rP)(t,x)$, therefore we have $Pr=rP$.
Similarly, we have $Tr=rT$, and $(PT)r=r(PT)$.
\\
All elements of $\mathcal{K}$ commute with all elements of $\mathcal{R}$.
\\
We deduce that both $\mathcal{K}$ and $\mathcal{R}$ are normal subgroups of $\mathcal{Q}$.

(b) We notice that a direct product of groups is a special case of a semidirect product of groups, and we apply Theorem \ref{th:SDPD}.

(c) Apply Theorem \ref{th:DPI}.
\end{proof}

Theorems \ref{th:GALK} and \ref{th:GALQ} should not be confused with the theorem in vigor in $E^4$ stating that the product of two indirect isometries is always a direct isometry.
The deep reason of this difference comes from the fact that not all rotations of $E^4$ are in $\mathcal{Q}$.

\begin{Theorem}\label{th:GALB}
(a) The isometry group $\mathcal{L} = \mathcal{B} \cup \mathcal{Q}$ is an outer semidirect product of $\mathcal{B}$ and $\mathcal{Q}$, \mbox{that is}, \mbox{$\mathcal{L} = \mathcal{B} \rtimes \mathcal{Q}$}.
\\
(b) The direct symmetries of $\mathcal{Q}$ are direct symmetries of $\mathcal{L}$.
\\
(c) The indirect symmetries of $\mathcal{Q}$ are indirect symmetries of $\mathcal{L}$.
\end{Theorem}

\begin{proof}
(a) $\mathcal{B}$ is an Abelian group (see Section \ref{sec:GAL}) and so it is a subgroup of $\mathcal{L}$.
$\mathcal{Q}$ is a group (\mbox{Theorem \ref{th:GALQ}}) and so it is a subgroup of $\mathcal{L}$.
The intersection of $\mathcal{B}$ and $\mathcal{Q}$ is the identity. 

We denote by $q_\Omega$ an element of $\mathcal{Q}$ defined by its parameter $\Omega$.
These two entities must not be confused: $q_\Omega$ acts on the vectors $(t,x) \in \mathbb{R}^4$, while $\Omega$ acts on the spatial component $x \in \mathbb{R}^3$ of $(t,x)$ and returns its image in $\mathbb{R}^3$.

For all $b_v \in \mathcal{B}$ of direction $v \in \mathbb{R}^3$, and for all $q_\Omega \in \mathcal{Q}$ and for all vectors $(t,x)$, we observe that $(q_\Omega b_v q_\Omega^{-1})(t,x)=(t,x+t \Omega v))$.
Then, $t$ being a real, we deduce that $q_\Omega b_v q_\Omega^{-1}$ is a boost of direction $t \Omega v$, and it follows that $\mathcal{B}$ is a normal subgroup of $\mathcal{L}$.

(b) Apply Theorem \ref{th:SDPD}.

(c) Apply Theorem \ref{th:SDPI}.
\end{proof}

\begin{Theorem}\label{th:GAL}
(a) The full isometry group $GAL(1,3)$ is an outer semidirect product of $\mathcal{T}$ and $\mathcal{L}$, that is, $GAL(1,3) = \mathcal{T} \rtimes \mathcal{L}$.
\\
(b) The direct symmetries of $\mathcal{L}$ are direct symmetries of $GAL(1,3)$.
\\
(c) The indirect symmetries of $\mathcal{L}$ are indirect symmetries of $GAL(1,3)$.
\end{Theorem}

\begin{proof}
(a) $\mathcal{T}$ is a group (see Section \ref{sec:GAL}) and so it is a subgroup of $GAL(1,3)$, $\mathcal{L}$ is a group (\mbox{Theorem \ref{th:GALB}}) and so it is a subgroup of $GAL(1,3)$.

Remembering that $\mathcal{B} \cap \mathcal{Q}$ is the identity, we need to know if the intersection of $\mathcal{T}$ and $\mathcal{L}=\mathcal{B} \rtimes \mathcal{Q}$ is the identity or not.

Given a translation $\tau \in \mathcal{T}$ of parameter $(s,z)$, $s$ being a time component and $z$ a spatial component, the image $\tau (t,x)$ of a vector $(t,x)$ is $(t+s,x+z)$.
The parameter $(s,z)$ must depend neither on $t$ nor on $x$.

We look at $\mathcal{T} \cap \mathcal{B}$ and we assume that it exists $b_v \in \mathcal{B}$ such that $\tau (t,x) = b_v (t,x)$, that is, $(t+s,x+z)=(t,x+tv)$.
It means that $s=0$ and that $z=tv$, which is not a constant unless $z=0$ and $v=0$, and therefore $\mathcal{T} \cap \mathcal{B}$ is the identity.

We look at $\mathcal{T} \cap \mathcal{Q}$ and we assume that it exists $q_\Omega \in \mathcal{Q}$ such that $\tau (t,x) = q_\Omega (t,x)$, that is, $(t+s,x+z)=(t,\Omega x)$.
It means that $s=0$ and that $z=\Omega x - x$, which is not a constant unless $z=0$ and $\Omega$ is the identity.

Finally, we deduce that $\mathcal{T} \cap \mathcal{L}$ is indeed the identity.

Then we look at the action of $l \tau l^{-1}$ on a vector $(t,x)$, where $l \in \mathcal{L}$ and $\tau \in \mathcal{T}$.
From Theorem \ref{th:GALB}, $\mathcal{L} = \mathcal{B} \rtimes \mathcal{Q}$, so we consider two cases: $l \in \mathcal{B}$ and $l \in \mathcal{Q}$.
The space and time translation $\tau$ has for parameter $(s,z)$, where $s \in \mathbb{R}$ is the time component of $\tau$ and $z \in \mathbb{R}^3$ is the spatial component of $\tau$.

Let $b_v \in \mathcal{B}$ be a boost of direction $v \in \mathbb{R}^3$.
It transforms a vector $(t,x)$ into $(t,x+tv)$, and $b_v^{-1}$ transforms $(t,x)$ into $(t,x-tv)$.
Then, $\tau b_v^{-1} (t,x) = (t+s,x-tv+z)$, and $b_v \tau b_v^{-1} (t,x) = (t+s,(x-tv+z)+(t+s)v)$.
Thus, for all $\tau \in \mathcal{T}$ and all $b_v \in \mathcal{B}$, $b_v \tau b_v^{-1}$ is a translation of parameter $(s,z+sv)$.

We denote by $q_\Omega$ be an element of $\mathcal{Q}$ defined by its parameter $\Omega$.
For all $\tau \in \mathcal{T}$ and for all $q_\Omega \in \mathcal{Q}$, and for all vectors $(t,x)$, we observe that $(q_\Omega \tau q_\Omega^{-1})(t,x)=(t+s,x+\Omega z)$.
We deduce that $q_\Omega \tau q_\Omega^{-1}$ is a translation of parameter $(s,\Omega z)$.

Finally, it appears that $\mathcal{L}$ is a normal subgroup of $GAL(1,3)$.

(b) Apply Theorem \ref{th:SDPD}.

(c) Apply Theorem \ref{th:SDPI}.
\end{proof}

\section{The cone rotating in the Newton-Cartan spacetime}\label{sec:CONE}

To illustrate what is chirality in the Newton-Cartan spacetime, we consider the example of an homogeneous conical solid rotating on its symmetry axis.
This object was considered as a model for rotating molecules, and its chirality was subject during several years to a terminological controversy between two authors. 
These authors noticed that submitting the rotating cone to a suitable rotation followed by parity inversion $P$ and time reversal $T$ restored the original rotating come.
The question was to decide if this cone is a direct-symmetric object or an achiral object.

The first author considered that the rotating cone is a chiroid \cite{Mislow1999} (the term chiroid was defined in \cite{Whyte1958}), and the other one considered that the rotating cone has a false chirality, according to his own terminology \cite{Barron1986a,Barron2013}.
While we agree that terminological issues are important \cite{Petitjean2021c}, we do not enter in this controversy, because it is based on different definitions of chirality, so that both authors may be right in their analysis.
\mbox{We just} notice that the authors did not specify which model of spacetime they considered, and that they their definitions of chirality were physical ones rather than mathematical ones.
We give below our own analysis, but it is not intended to decide which of the authors was right, it is just to illustrate our current approach to chirality.

We consider the cone as a free rigid body.
Its motion can be described by the translation of its center of mass $x_c$ and by the rotation about $x_c$.
The object $(\rho,\nu,\eta)$ representing the cone consists of a scalar field $\rho$, which is the volumetric mass density, and of two vector fields, the linear velocity $\eta$ of the mass center $x_c$, and the angular momentum $\eta$ about $x_c$.
We assume that the conical solid is homogeneous, that is, $\rho$ is proportional to the indicator function of the conical domain, and thus $\rho$ has a symmetry axis.
The motion of the cone is a rotation on its symmetry axis at constant speed, which means that $\nu=0$ and that $\eta$ is a constant vector in the direction of the symmetry axis.
All planes containg the symmetry axis are symmetry planes, and $\rho$ is insensitive to reflections in these planes.

For clarity, we assume that one of the planes containing the symmetry axis is orthogonal to the spatial vector $(1,1,1)$.
Obviously, applying both isometries $P$ and $T$ together leaves the object $(\rho,\nu,\eta)$ invariant.
From Theorems \ref{th:GALK}-\ref{th:GAL}, we know that $PT$ is an indirect isometry of $GAL(1,3)$.
It is in the symmetry group of $(\rho,\nu,\eta)$, from which we conclude that the rotating cone is achiral.

It would have been impossible to conclude about the chirality or achirality of the rotating cone or of any other moving body without specifying the nature of the spacetime and classifying the isometries as direct or indirect.

\section{Conclusions}\label{sec:CONC}

Given a space, its isometry group depends only on which kind of invariance is considered: invariance of distances (metric space), of intervals (quadratic space), or else.
The knowledge of this isometry group contains all information required to classify the isometries as direct and indirect.
Then, it can be decided if an object defined in this space is direct-symmetric or not, and if it is chiral or not.

Our approach was based on an extension of the usual chirality concept in Euclidean spaces to spacetime \cite{Petitjean2019,Petitjean2021b}, but it should not be confused with the chirality concept specific to quantum field theory \cite{Petitjean2020b}.

A physical system may receive several mathematical models.
So, a \mbox{mathematical} model of a physical system must not be confused with the physical system itself.
\mbox{As recalled} in \cite{Petitjean2021c}, it is crucial to understand that mathematical constructs are \mbox{models} of real physical situations and that a mathematical model of symmetry is a simplified image in our mind of some physical situation in which we would like to see symmetry.
Unending controversies about chirality may arise if terminological ambiguities occur.

\vspace{6pt}

\section*{Acknowledgements} I am highly grateful to Prof. Emil Moln\'ar for supporting me to publish this work.

\renewcommand{\refname}{References}

\end{document}